\documentclass{article}
\usepackage{amsmath}
\usepackage{amssymb}
\usepackage[ruled,lined]{algorithm2e}
\usepackage{enumerate}
\usepackage{graphicx}
\usepackage{color}

\newtheorem{lemma}{Lemma:}
\newenvironment{proof}{{\bf Proof:}}{\hfill {\fbox{}}}

\newtheorem{ex}{Example:}
\usepackage{lipsum}
\usepackage{tabularx}
\usepackage[a4paper]{geometry}

\newtheorem{defi}{Definition}

\newcolumntype{r}{>{\raggedright\arraybackslash}m{20mm}}

\begin{document}
\title{An Optimal Algorithm for Range Search on Multidimensional Points}
\date{\vspace{-5ex}}

\author{T.Hema \thanks{Email: hema@cs.annauniv.edu} and K.S. Easwarakumar\thanks{Corresponding Author.Email: easwara@cs.annauniv.edu} \\Department of Computer Science \& Engineering\\ Anna University, Chennai 600 025, INDIA.}
\maketitle

\begin{abstract} This paper proposes an efficient and novel method to address range search on multidimensional points in $\theta(t)$ time, where $t$ is the number of points reported in $\Re^k$ space. This is accomplished by introducing a new data structure, called BITS \textit{k}d-tree. This structure also supports fast updation that takes $\theta(1)$ time for insertion and $O(\log n)$ time for deletion. The earlier best known algorithm for this problem is $O(\log^k n+t)$ time \cite{b3,mb} in the pointer machine model.  

\end{abstract}

\section{Introduction}
\label{intro}
\textit{k}d-trees introduced by J.L.Bentley \cite{b5,b8} are multidimensional binary search trees commonly used for storing $k$ dimensional points. They are also used to perform search operations such as exact match, partial match and range queries. Range queries are mostly used in $GIS$ applications to locate cities within a certain region in a map. Similarly, in the geometrical view of a database, one can use orthogonal range search to perform a query. 
Generally, \textit{k}d-trees with $n$ nodes have a height $n$ and hence the complexity for insertion and search are high. Although many multi-dimensional search structures are found in the literature \cite{pk,bcko,ps,h2,h1} they differ from the standard \textit{k}d-trees mainly in the space-partitioning methods used. Recall that a $2$-d tree stores two-dimensional point data of the form $(x,y)$. A $2$-d tree splits primarily on the $x$ coordinate of a point at even level and then on the corresponding $y$ coordinate at the odd level, and so on. Hence, the trees are unbalanced and are not efficient for search operations. Also, the worst case time complexity for range search on a $2$-$d$ tree is $O(\sqrt{n}+t)$, where $t$ is the number of points reported and for $k$ dimensions it is $O(n^{1-1/k}+t)$\cite{b5,lw}. In general, most of the \textit{k}d-tree variants get unbalanced when the data is clustered thereby affecting query operations.
\par 
$PR$ \textit{k}-d tree, Bucket $PR$ \textit{k}-d tree \cite{jo}, $PMR$ \textit{k}-d trees \cite{ns} and Path level compressed $PR$ \textit{k}-d trees\cite{nt} are some of the trie-based kd trees used to store point data. However, these trees are not always balanced, especially when the data is clustered. One of the dynamic versions of \textit{k}-d tree is the divided \textit{k}-d trees \cite{km} for which the range query time is $O(n^{1-1/k}\log^{1/k} n+t)$. 
\par The best known dynamically balanced tree uses bitwise interlaced data \cite{hth} over \textit{k}d-trees mapping $k$ dimensions to one dimension. Although their search time is $O(k(\log~n+t))$  for reporting $t$ points, bitwise interlacing leads to discarded areas during range search. In the case of squarish \textit{k}d-trees \cite{ldjc}, an $x$, $y$ discriminant is based on the longest side of rectangle enclosing the problem space instead of alternating the keys. Recently, hybrid versions of squarish \textit{k}d-tree, relaxed \textit{k}d-tree and median \textit{k}d-trees \cite{mpc} have overcome the problem of height balancing. An amortized worst case efficiency of range search for the hybrid squarish \textit{k}d-trees, relaxed and median trees for \textit{k}-dimensional 
partial match queries are $1.38628~log_{2}n$, $1.38629~log_{2}n$ and  $1.25766~log_{2}n$ respectively. Their experimental results match the aforementioned theoretical results, where they show that the hybrid median trees outperform the other variants. However, as far as query handling is concerned, these structures perform only partial match queries for two dimensions efficiently. The most recent work in the pointer machine model is an orthogonal range reporting data structure with $O(n(\log~n/\log log n)^d)$ space that address range queries in $O(\log~n(\log n/log~log~n)$ time, where $d\geq 4$ \cite{af}. 
\par
Range trees of Bentley and Maurer \cite{b8,b3} are yet another class of balanced binary search trees used for rectangular range search which showed improvement in the  query time of $O(\log^k n+t)$ over $O(n^{1-1/k}+t)$ of $k$d-trees, where $k$ is the dimension for a set of $n$ points and $k$ is the number of reported points. This was later improved to $O(\log^{k-1}n+t)$ using fractional cascading in layered range trees\cite{w} but the space requirements are relatively high of $O(n\log^{k-1}n)$. A \textit{k}d-Range $DSL$-tree performs $k$-dimensional range search in $O(\log^kn+t)$ time was proposed in \cite{mb}. 
\par 
Recently, Chan et.al \cite{clkm} have proposed two data structures for $2d$ orthogonal range search in the word $RAM$ model. The first structure takes $O(n~lg~lg~n)$ space and $O(lg~lg~n)$ query time. They show improved performance over previous results \cite{als} of which $O(n~lg^\epsilon~n)$ space and $O(lg~lg~n)$ query time, or with $O(n~lg~lg~n)$ space and $O(lg^2~lg~n)$ query time. The second data strucure is based on $O(n)$ space and answers queries in $O(lg^\epsilon~n)$ time that outperforms previous $O(n)$ space data structure  \cite{nek}, answers queries in $O(lg~n/lg~lg~n)$ time.
\par
Furthermore, they also propose an efficient data structure for $3$-$d$ orthogonal range reporting with $O(n~{lg^{1+\epsilon}}+n)$ space and $O(lg~lg~n+k)$ query time for points in rank space where $\epsilon > 0$. This improves their previous results \cite{ct} with $O(n~lg^2~n)$ space and $O(lg~lg~n+k)$ query time, or with $O(n~lg^{1+\epsilon}~n)$ space and $O(lg^2~lg~n+k)$ query time, where $k$ points are reported. Finally they have extended range search to higher dimensions also.
     
Since such range queries are common among multi-dimensional queries in database applications, we have mainly considered an orthogonal range search on multi-dimensional points. 
\section{Our Contributions}
In this work, we make use of the $BITS$-tree \cite{ekh}, a segment tree variant that performs stabbing and range queries
on segments efficiently in logarithmic time. Most importantly, the distribution of the data points (uniform or skewed) does
not affect the height of the $BITS$-tree and in turn facilitates faster search time. Here, we actually use the $BITS$-tree structure to store points related to each dimension and thereby form a multi-level tree, called $BITS$ \textit{k}d-tree. In addition, certain nodes of the $BITS$-tree associate to a variant of the trie data structure, called \textit{threaded trie}, to facilitate fetching a required node in constant time. Unlike \textit{k}-d trees, it does not associate co-ordinate axis, level wise, for comparison to locate or insert a point. Instead, the tree at the first level has nodes with a key on only distinct values of first co-ordinate of the points. Therefore, this tree corresponds to the one dimensional data. This tree is then augmented with another tree at second level and there in key values of the nodes associated with distinct first two co-ordinates of the points. In general, $i^{th}$ tree corresponds to the distinct first $i$ co-ordinates of the set of points given. Moreover, in each tree, the inorder sequence provides the sorted sequence.  
That is, $BITS$ \textit{k}-d trees is a multi-level tree, and its construction is illustrated in the subsequent sections.
\subsection{$BITS$-Trees}
Originally, the $BITS$-tree (\textbf{B}alanced \textbf{I}norder \textbf{T}hreaded \textbf{S}egment Tree) \cite{ekh} is a dynamic structure that stores segments, and also answers both stabbing and range queries efficiently. Unlike segment trees, it also permits insertion of segment with any interval range.
\begin{defi}
A $BITS$-tree is a height balanced two-way inorder-threaded binary tree $T$ that satisfies the following properties.
\begin{enumerate}[1.]
\itemsep=0pt
\parskip=0pt
\item Each node $v$ of $T$ is represented as $v([a,b],L)$, where $[a,b]$ is the range associated with the node $v$, and $L$ is the list of segments containing the range $[a,b]$, i.e if $[c,d] \in L$ then $[a,b]\subseteq[c,d]$.  
\item Given $v_1([a_1,b_1],L_1) \neq v_2([a_2,b_2],L_2)$, then
	\[
	 \left[a_1, b_1\right]\cap \left[a_2, b_2\right] = \left \{ \begin{array}{lcl} \left[b_1\right]~ if ~ a_2=b_1 \\
																										 \left[b_2\right] ~ if ~ a_1=b_2 \\
																									 	  \phi ~~ otherwise
                           				\end{array} \right .
\]
i.e ranges can either overlap only at end points or do not overlap at all.  
\vskip -1 em          
\item Suppose $v_1([a_1,b_1],L_1)$ appears before $v_2([a_2,b_2],L_2)$ in the inorder sequence, then $b_1 \leq a_2$. 
\item It has a special node, called $dummy$ node denoted by $D$, with range and list as $\phi$(empty). 
\item Suppose $v_1([a_1,b_1],L_1)$ and $v_n([a_n,b_n],L_n)$ are the first and last nodes of the inorder sequence respectively, then $InPred(v_1)=InSucc(v_n)=D$, and the range, say $[a,b]$, of any node contained in $[a_1,b_n]$, i.e $[a,b] \subseteq [a_1,b_n]$. 
\end{enumerate}
\end{defi}
Here, the functions InPred() and InSucc() respectively returns inorder predecessor and successor.
A sample $BITS$-tree is shown in Figure \ref{f1}. 
\begin{figure*}
\centering
\includegraphics[width=\textwidth]{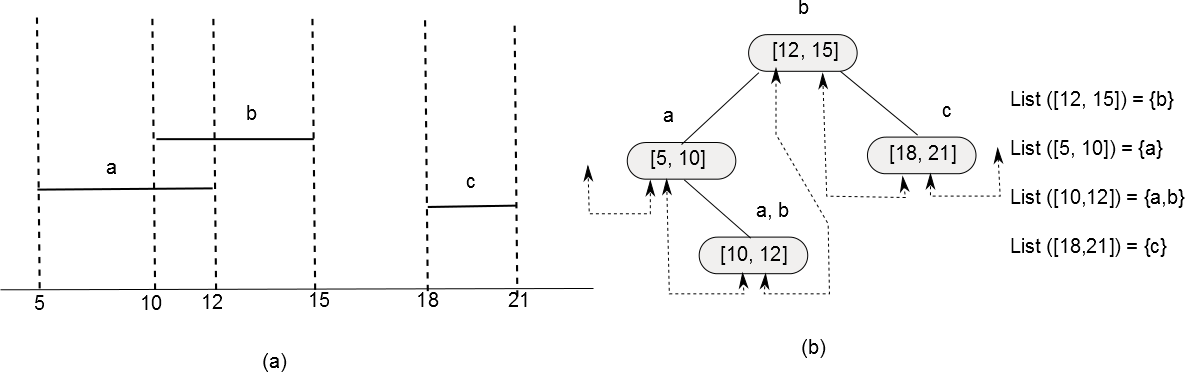}
\caption{(a) Set of segments.(b) $BITS$-Tree for the given segments.}
\label{f1}
\end{figure*}\\Note that the dangling threads actually point to a \textit{dummy} node, which is not shown in the figure.
\par
The $BITS$-tree is originally developed for storing segments, but we use this for a different purpose of storing points. Thus, we modify this structure to suit our requirement as described below.
\begin{enumerate}[1.]
\itemsep=0pt
\parskip=0pt
\item Each node $v([a,b],L)$ is replaced by $v(p,L',T)$, where $p$ is a point in $\Re^k$, $k \ge 1$, and $L'$ is a pointer to the list of collinear points in dimension $k+1$, having $p$ for the first $k$ co-ordinates. However, this list is maintained in the tree at the next level, which is described in section \ref{tt}. Now, $T$ is either null or a pointer to a \textit{threaded trie}, which is elaborated in the next section.
\item For any two points $p_1$ and $p_2$ stored in a tree, $p_1\neq p_2$.
\item Suppose $v_1(p_1,L_1',T_1)$ appears before $v_2(p_2,L_2',T_2)$ in the inorder sequence, then $p_1 < p_2$ as per the following definition.  
\end{enumerate}
\begin{defi}
Let $p_1=(x_1^1, x_2^1,\ldots x_k^1)$ and $p_2=(x_1^2, x_2^2,\ldots x_k^2)$ be two points in a $k$-dimensional space, then
\begin{enumerate}[a.]
\itemsep=0pt
\parskip=0pt
\vskip -1em
\item $p_1=p_2$ implies $x_j^1=x_j^2$ for each $j$=1, 2,\ldots k.
\item $p_1 < (or >) p_2$ implies $head(p_1,j)=head(p_2,j)$ and $x_{j+1}^1$ $<$ (or \textgreater) $x_{j+1}^2$ for some $j$.
\end{enumerate}
\label{d4}
\end{defi}
\par
In the subsequent sections, for better clarity, we use hyphen($-$) for a certain parameter of a node to denote that the particular parameter is irrelevant with respect to the context. For instance, $(p,-,T)$ denotes that the list contents are irrelevant for that point $p$ at this time. 
\subsection{Threaded Trie}
\begin{figure}[!ht]
\centering
\includegraphics[width=1.5in,height=2.1in]{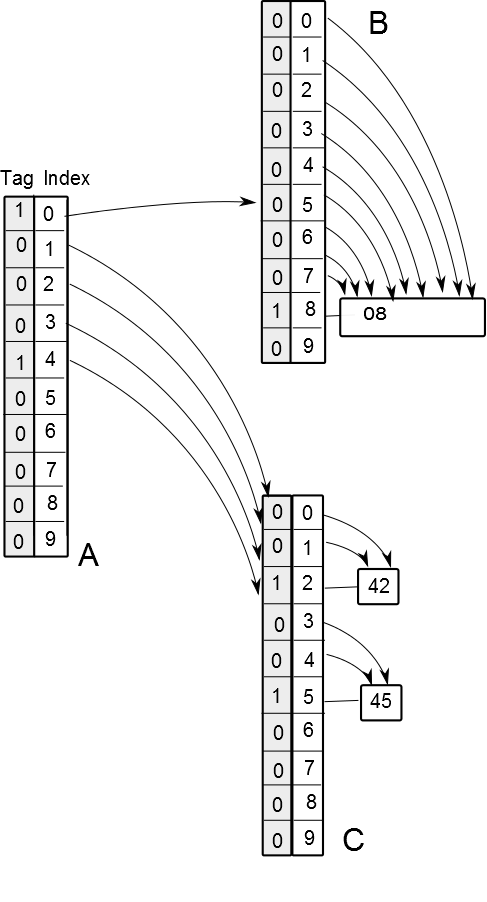}
\caption{A sample threaded trie}
\label{bitsfig1}
\end{figure}
Threaded tries are variants of tries that consists of two types of nodes viz. trie node and data node. For instance, in Figure \ref{bitsfig1}, $A$, $B$ and $C$ are trie nodes and the rest are data nodes. Unlike in tries, the trie node here does not have a field for blank ($\not{\hspace{-0.1cm}b}$). However, each of these trie nodes contain two segments. One is the index pointer and the other is a tag value, which is either $0$ or $1$, where $0$ denotes the corresponding index point in a thread and otherwise it will be $1$. Here, all null pointers are replaced by threaded pointers, which point to the next valid node, if one exists. For instance, the thread pointers of $1$, $2$, $3$ of node $A$ points to the node $C$, as this is the next valid node. Similarly, thread pointers of $0$ and $1$, in $C$ points to the data node $42$. Note here that ordering on the nodes provides the sorted sequence. Also, data nodes appear at the same level. This is accomplished by having uniform width for all data. For instance, the data $8$ is treated as $08$ in Figure \ref{bitsfig1}. 
\label{tt}
\subsection{Construction of Multi-level $BITS$-Tree}
\label{cons}
Multi-level $BITS$-trees are constructed using a collection of $BITS$-trees one at each level, and interlinking the trees of two consecutive levels in a specified manner, which are due to the following definitions. These multi-level $BITS$-trees are termed here as $BITS$-\textit{k}d trees.
\begin{defi}
\label{d1}
Given a point $p = (x_1, x_2, \ldots  x_k)$ and an integer $l \le k$, the head of $p$ and tail of $p$ are defined
respectively as $head(p,l)=(x_1, x_2, \ldots x_l)$ and $tail(p,l)=(x_{k-l+1}, x_{k-l+2},\ldots x_k)$. Also, 
having $(head(p,l),y_1,y_2,\ldots y_m)=(x_1, x_2, \ldots x_l, y_1, y_2, \ldots y_m)$, leads to $(head(p,l),\\tail(p,k-l))=p$.
\end{defi}
\begin{defi}
\label{d2}
\sloppy Given $S$ as the set of points in $\Re^k$ and $\left| S \right|$=n, the set $S_1$ is defined as, $S_1=\bigcup_{i=1}^{n} \{(x) ~ | ~ p \in S  ~  and  ~  head(p,1)=(x) \}$. That is, $S_1$ is the set of distinct $x$ values of the points in $S$. In general, $S_j = \bigcup_{i=1}^n \left\{(x_1, x_2, \ldots x_j) ~ | ~ p \in S\right .$ and  $\left . head(p,j) = (x_1, x_2, \ldots x_j) \right\}$, where $1 \leq j \leq k$.
\end{defi}
\begin{defi}
For a point $p=\left(x_1,x_2,\ldots x_j\right)$ in $S_j$, the term $x_j$ is said to be the dimensional value of $p$ as the set of points in $S_j$ is used to construct $j^{th}$ level $BITS$-tree.
\end{defi}
\begin{figure*}[!t]
\centering
\includegraphics[width=.9\textwidth]{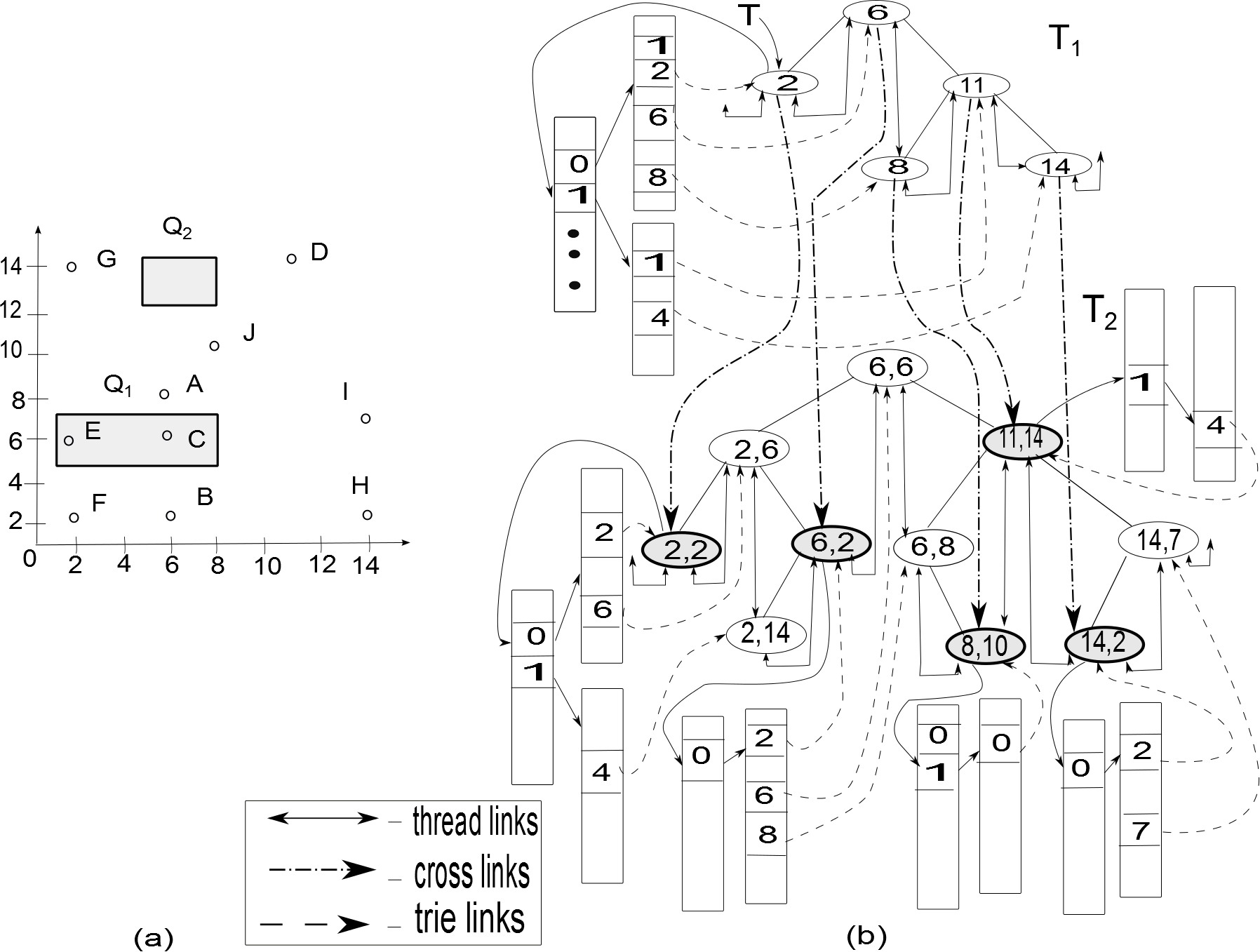}
\caption{A $BITS$ 2d-tree: (a) Spatial representation of  points. (b) $BITS$-2d tree for points shown in (a).}
\label{bits2d}
\end{figure*}

\begin{defi}
\label{d3}
A $BITS$ \textit{k}d-tree is a multi-level tree, which is constructed as follows. 
\begin{enumerate}[1.]
\itemsep=0pt
\parskip=0pt
\item Create separate $BITS$ trees, $T_j$ for each $S_j$, $1 \leq j \leq k$. 
\item Let $X_j=(x_1,x_2,\ldots x_j)$. Now, for each node, say $v_j=\left(X_j,L,-\right)$, of $T_j$, $1 \leq j < k$, the list $L$ points to the node $v_{j+1}=((X_j,x^{\prime}_{j+1}),-,-)~of~T_{j+1}$,\\~where~$x_{j+1}^{\prime} = min \{x_{j+1}~| ~ head(p,j) = X_j$~and~$(head(p,j), x_{j+1}) \in S_{j+1}\}$. We term these links as \textit{cross links}, and the node $v_{j+1}$ as a cross link node in $T_{j+1}$.
\item In $T_1$, there is only one cross link node, which is the first node in the in order sequence, and the tree pointer always points to this node.
\item For each node $v=((X_{j-1},x_j^{\prime}),-,T)$ in $T_j$, $1 \leq j \leq k$, $T$ is pointer to the threaded tree if $v$ is a cross link node, and otherwise $T$ is set to be null.
\item For every cross link node $v_j=((X_{j-1},x_j^{\prime}),-,T)$ in $T_j$, the data node of $T$ for a key, say $k^{\prime}$, points to the node $((X_{j-1},k^{\prime}),-,-)$ in $T_j$. That is, $T$ provides links to the nodes in $\left\{((X_{j-1},x_j),-,-)|(X_{j-1},x_j) \in S_j\right\}$ and these links are termed as \textit{trie links}.    
\end{enumerate}
\end{defi}
A $BITS$ 2d-tree for the sample points in Figure \ref{bits2d}(a) is shown in Figure \ref{bits2d}(b).
Since, $BITS$ \textit{k}d-trees are multi-level trees with binary inorder threaded search trees at each level, the height of the trees at each level is $O(\log~n)$. Also, each node in $T_{i-1}$ has a \textit{cross link} to a node in $T_{i}$, which has the least value for the $i^{th}$ co-ordinate with respect to the head value of the node in $T_{i-1}$. Note here that at least one such point exists. This link is useful to locate a list of collinear points in the $i^{th}$ dimension, associated with a point in $T_{i-1}$. Also, the trie links are useful to locate a point in a given range window in constant time. The cross link and trie links also make the structure much suitable to address range queries efficiently.
\par Normally, \textit{k}d-trees perform insertion by a simple comparison between the respective co-ordinates at each level. However, deletion is  tedious due to \textit{candidate replacement}. This is because, candidate for replacement can be anywhere in the subtree. Also, it requires a little more work when the right subtree is empty. Now, to find a candidate for replacement, it is required to find the smallest element from the left subtree to avoid violation of the basic rules of \textit{k}d-trees and then it is required to perform a swap of left and right subtrees, as many possible candidate keys exist in the left subtree. To handle such a situation, we make use of a collection of $BITS$-trees, one for each dimension. Here, deleting a point may or may not require a replacement, but if so, it is only the inorder successor and that can be located in $\theta(1)$ time as inorder links exist for each node. Also, the \textit{cross links} that exist between two consecutive levels, practically provide a faster search on next level trees.
\par Another advantage of this structure is that when a node is pruned out at a particular level, it need not be considered in the subsequent levels. That is, nodes that have \textit{head} values as these will be ignored in the subsequent levels. To the best of our knowledge, there is no such structure using multi-levels of balanced binary search trees, with two-way threads introduced in this work, for storing point data and to perform range search efficiently.  

\subsection{$2d$-Range Search for Window Query}
\label{2drs}
Given a rectangular range in the form of a window, a range query finds all points lying within this window. Let $[x_1: x_2]~\times~[y_1: y_2]$ be a given query range. First, we use a trie stored in the first node, which is the only cross link node, in the tree at level $1$ to find the smallest point larger than or equal to $x_1$ in $S_1$. The non-existence of such a point is determined from the trie itself. On the other hand, once such a point $p$ is located, subsequent points that fall within $[x_1:x_2]$ can be determined using the inorder threads as the inorder sequence is in sorted order. Let us say that the reported set of points as $S^{\prime}$. However, if the dimensional value of the point $p$ is greater than $x_2$, it implies absence of required candidates.
\par
Now, using cross links of the node in $T_1$, that corresponds to each point in $S^{\prime}$, further search is performed at $T_2$ in a similar fashion. Note that each cross link node in $T_2$ has a trie structure that supports quick access to a node in $T_2$, where the dimensional value is in $[y_1:y_2]$. In case, the dimensional value of the cross link node is within $[y_1:y_2]$, the respective trie structure need not be looked into, instead the inorder threads are used to find the remaining candidates. 
\begin{ex}
For instance, let us consider Figure \ref{bits2d} with search range $[1:8]\times[5:7]$. First, we use the only cross link node present in $T_1$. As its dimension value, ie. $2$ lies within the range $[1:8]$, we do not use the respective trie. Instead, we use the inorder threads to identify the candidate points, which are 2,6 and 8. Now, for each of these candidates, further search is continued respectively from $(2,2)$,$(6,2)$ and $(8,10)$ in $T_2$, as these are the corresponding cross link nodes. Now, by looking at the tries of these cross link nodes we find a point whose dimensional value is the smallest one is $[5:7]$. Thus, tries of $(2,2)$ yields $(2,6)$, $(6,2)$ yields $(6,6)$ and $(8,10)$ yields nothing. Further, by performing inorder traversal from $(2,6)$ and $(6,6)$, the final reported points for $Q_1$ are $E(2,6)$ and $C(6,6)$. Also, for $Q_2$ i.e, $([5:8]\times[12:14])$, no points will be reported.
\end{ex}
\par Notice that one can stop the search at $T_1$ without traversing $T_2$ if there is no candidate node in $T_1$ within the given range. This is also applicable in $k$-d trees because if there is no candidate node in the higher tree, the lower level trees need not be searched. Thus, this structure prunes the search in some cases and thereby practically reduces the time for reporting a query.
\subsection{\textit{k}-d Range Search}
A range search on $k$-dimensional points can be performed by extending the search on $T_3, T_4, \ldots  T_k$, similar to that of $T_2$ as in the case of 2d range search. However in $T_1$ and $T_2$, we need to perform the search as described for 2d range search. That is, when we take the query range as $[x_1:x_1^{\prime}]\times[x_2:x_2^{\prime}]\times \ldots [x_k:x_k^{\prime}]$, the search is performed to find candidates within the range of $[x_1:x_1^\prime]$ in $T_1$, $[x_2:x_2^{\prime}]$ in $T_2$, $[x_3:x_3^{\prime}]$ in $T_3$, and so on. Finally, the points reported from $T_k$ will be in $Q$.
It is important to note that the search requires comparison of keys within the given range of the particular co-ordinate dimension in each of $T_1,T_2,\ldots T_k$. This simplifies subsequent searches at the next level. 
\section{Implementation Details}
\subsection{Two Dimensions}
\label{td}
Given a set of two dimensional points in $\Re^2$, a two-level tree ($BITS$2d-tree) is constructed in $O(n)$ time as a point may require at most two insertions, one at $T_1$ and the other at $T_2$. But the position at which insertion is to be made in $T_1$ and $T_2$ could be determined in constant time as described in the proof of Lemma \ref{ins}. Thus, to insert $n$ nodes requires $O(n)$ time. Also, it may be required to create a \textit{cross link} for each node of $T_1$ in the case of $BITS$ 2d-tree. Since, $T_1$ cannot have more than $n$ points, the number of cross links created cannot exceed $n$. Also, the number of trie links created cannot exceed the number of nodes in $T_1$ and $T_2$, which is $O(n)$. Moreover, construction of a trie requires only constant time as the height of the trie is constant due to fixed size of the key. Thus, all these factors lie within $O(\log n)$ for each insertion.
\par
Regarding space requirements in a $BITS$ 2d-tree, it is $O(n)$, as the second tree is the one that contains all the $n$ points, and fewer or equal number of points in the first tree. Also, the number of trie nodes is $O(n)$ as the height of a trie is constant which is due to the size of(number of digits) of the key. Thus, we obtain the following lemma.
\begin{lemma}Construction of $BITS$ 2d-tree for $n$ points requires $O(n)$ time and $O(n)$ space.
\end{lemma}
Now, searching a candidate node in $T_1$ is done through the trie in $T_1$ and that requires only constant time as the height of the trie is fixed. Once such a point is identified, subsequent points are identified through inorder threads. Thus. for identifying candidate points, it takes only $\theta(t_1)$ time, if there are $t_1$ candidate points in $T_1$. Now, using cross links of each of these nodes, we can locate the required tries in constant time and further search is to be done in a similar fashion as described earlier. Thus, it leads to the following lemma.
\begin{lemma}
\sloppy Range search for window query using $BITS$ 2d-tree can be addressed in $\theta(t)$ time,
where $t$ stands for the number of points reported. 
\end{lemma}
\subsection{Higher Dimensions}
\label{hd}
A straight forward extension of $BITS$ 2d-tree to $k$ dimensions is made easy by connecting (cross links) to the corresponding nodes in the tree at next level. Unlike range trees \cite{b7} which build another range tree at a given node from the main tree, we maintain the trees $T_1, T_2, \ldots T_k$, dimension-wise such that the inorder traversal provides an ordered sequence of points stored in the tree. This definitely reduces the overall time taken for range search across $k$ dimensions. As described in the previous section, the time required to find a candidate point in any $T_i$, $1\leq i \leq k$, is only a constant. Thus, it leads to the following lemma.
\begin{lemma}
Let $S$ be a set of points in $k$-dimensional space, $k\geq 1$. A range search on $BITS$~\textit{k}d-tree reports all points that lie within the rectangular query range in $\theta(t)$ time, where $t$ is the number of points reported.
\label{rangeq}
\end{lemma} 
\begin{lemma}
\label{t1}
Given a set of $n$ points, a $BITS$~\textit{k}d-tree can be constructed in $O(n)$ time and $O(n)$ space.
\end{lemma}  
\begin{proof}
Since we construct $T_1, T_2 \ldots T_k$, such that $T_k$ at level $k$ has at most $n$ nodes, it follows that $N\left(T_i\right) \le N\left(T_{i+1}\right)$, $1 \le i < k$ and $N(T_k)$=$n$, where $N\left(T_i\right)$ is the number of nodes in $T_i$. Note that levels correspond to dimensions and hence may be used interchangeably. Also, the number of trie nodes is $O(n)$ as its height is constant. Therefore for $k$ levels, a $BITS$ \textit{k}d-tree uses $O(n)$ storage in the worst case as $k$ is a constant.  
Now, construction of $BITS$ \textit{k}d-tree is considered as a sequence of insertions. Each insertion, may or may not alter $T_i$, $1 \le i \le k$, a $BITS$ tree of a particular level. However, if a $BITS$-tree $T_j$ is altered, due to insertion, all trees $T_{j+1},T_{j+2}, \ldots T_k$ will be altered. Let $j$ be the least index such that the tree $T_j$ is altered. Thus, for $T_1, T_2,\ldots T_{j-1}$, with trie links and cross links, one can determine that the required values are already stored in those trees within constant time. Now, from a particular cross link in $T_{j-1}$ followed by a trie link in $T_j$, one can find a position for the new value in $T_j$. This requires only constant time. Then, while inserting the value if the tree is unbalanced, atmost one rotation is required to balance the tree. So, for $T_j$ too, it requires constant time. Let $n_j$ be the new node inserted in $T_j$. Now, by taking cross link of inorder successor of $n_j$, one can determine the position of the new node in $T_{j+1}$, and that as inorder predecessor of cross link node of inorder successor of $n_j$. This new node in $T_{j+1}$ need to have a trie, which again be created in constant time. Then, the process is to be continued for $T_{j+2} \ldots T_k$. Here, updation in each $T_i, 1\leq i \leq k$, takes only constant time and hence each insertion takes $\theta(1)$ time. So, construction of $BITS$ \textit{k}d-tree for $n$ points requires $O(n)$ time.
\end{proof}
\begin{table*}[t]
\caption{Theoretical comparison of \textit{k}d-trees, divided \textit{k}-d trees, range trees, \textit{k}-d Range~DSL-trees, layered range trees and the proposed $BITS$ \textit{k}d-trees.}
\label{tc}
\begin{center}
\footnotesize
\resizebox{\textwidth}{!} {%
\begin{tabular}{|l|l|l|l|l|}
\hline
\textbf{Description} & \textbf{Storage}  & \textbf{Construction} & \textbf{Update} & \textbf{Range Search} \\ \hline
\textit{k}d-Trees \cite{b5} &  $O(n)$ & $O(n \log n)$ & $O(\log ^k n)$ & $O(n^{1-1/k}+t)$\\ \hline
Divided \textit{k}-d trees \cite{km} & $O(n)$ & $O(n\log n)$ & $O(\log ^{k-1} n)$ & $O(n^{1-1/k}\log ^{1/k} n+t)$  \\ \hline
Range Trees\cite{b7} & $O(n \log ^{k-1} n)$ & $O(n \log^{k-1} n)$ & $O(\log ^ k n)$ & $O(\log^k n+t)$\\ \hline
\textit{k}d-Range DSL-Trees \cite{mb} & $O(n\log^{k-1} n)$ &$O(n\log^k n)$ & $O(n\log^{k-1} n)$  & $O(\log^k n+t)$\\ \hline
Layered Range Trees\cite{w} & $O(n \log ^{k-1} n)$ & $O(n \log^{k-1} n)$ & $O(\log ^ k n)$ & $O(\log^{k-1} n+t)$\\ \hline
BITS \textit{k}d-Trees & $O(n)$ & $O(n)$ & Ins.~$\theta(1)$~Del.~$O(\log n)$  & $\theta(t)$ \\ \hline
\end{tabular}%
}
\tiny
\noindent
\begin{tabularx}{\textwidth}{X}
$n$-number of points, $k$-dimensions, $t$-number of points reported.
\end{tabularx}
\end{center}
\normalsize
\end{table*}

\begin{lemma}
\label{ins}
Insertion and Deletion of a point in a $BITS$~\textit{k}d-tree can be respectively done in $\theta(1)$ and $O(\log n)$ time.\end{lemma}  
\begin{proof}
As per the description given in the proof of Lemma \ref{t1}, insertion of a point in BITS \textit{k}d-tree takes only $\theta(1)$ time. But for deletion, finding a node to be removed from a $BITS$-tree requires only constant time. However, if that node is not a leaf node a cascading replacement with inorder successor is required until reaching a leaf node to be removed physically. Certainly, the number of such replacements to be done cannot exceed $O(\log n)$. After that it may require a sequence of rotations on the path from the physically removed leaf to the root, and that too in at most $O(\log n)$ rotations. So, deletion of a point in $BITS$ \textit{k}d-tree requires $O(\log n)$ time.
\end{proof}
\section{Performance} 

Table \ref{tc} summarizes the performance of \textit{k}d-trees, divided \textit{k}-d trees, range trees, \textit{k}d-range $DSL$-trees and the $BITS $\textit{k}d- tree proposed in this work. Furthermore, our theoretical comparison of the $BITS$ \textit{k}d-tree is made with \textit{k}d-trees adapted for internal memory(pointer machine model) and not with any of the other bulk loading \textit{k}d-trees($RAM$ model). The results give an $\theta(t)$ query time using the $BITS$ \textit{k}d-tree that shows a reduction in time as compared to the existing bounds. Since we try to capitalize on the efficiency of balanced search trees at all the levels by using \textit{cross links} and \textit{trie links}, we ensure that the number of nodes visited during a range query is considerably reduced in $BITS$ \textit{k}d-tree. Observe that the storage is increased from $O(n)$ in \textit{k}d-trees to $O(n \log ^{k-1} n)$ in range trees while $BITS$ \textit{k}d-tree still maintains an $O(n)$. Notice that the update time for BITS\textit{k}d-tree has been reduced considerably. To summarize, although the storage requirements of $BITS$ \textit{k}d-tree are comparable to \textit{k}-d trees, divided \textit{k}-d trees, the construction and update time are improved considerably. Moreover, the overall query time is improved to $\theta(t)$ time where $t$ is the number of points reported as it prunes points falling outside the query region for each dimension. 
\vskip -2em
\section{Conclusion}
A $BITS$ \textit{k}d-tree for storing $k$-dimensional points having update and query operations efficiently than $kd$-trees is proposed. The main advantage of this tree is that it effectively handles the collinear points. As a result, number of nodes visited during search is much less compared to other \textit{k}d-tree variants that are either not height balanced or update operation is complex. In the case of height balanced \textit{k}d-trees, having better search efficiency, insertion is tedious. A \textit{k}-d range $DSL$ tree gives a logarithmic amortized worst case search time with efficient updates mainly for partial match queries and not for window queries. In $BITS$ \textit{k}d-tree, overall insertion time is $\theta(1)$. Moreover, points can be dynamically updated at each level. Since co-ordinate dimensions at each level are distributed and using threaded tries, we quickly find points falling within the query range. Also, points falling above and below the search range are pruned efficiently using cross links to the next level and inorder threads similar to the $BITS$-tree. In addition, \textit{threaded tries} introduced in this work link the node, having cross link, by means of trie links to find the points within the given range in constant time. Therefore, range search for points in a rectangular region using $BITS$~\textit{k}-d tree takes $\theta(t)$ time where $t$ is the number of points reported, and therefore the logarithmic factor in earlier worst case bounds is reduced. Hence it is definitely a remarkable improvement over $O(n^{1-1/k}+t)$ of \textit{k}d-trees and $O(\log^k n+t)$ time of \textit{k}-d range $DSL$ trees.
\bibliographystyle{plain}
\bibliography{BITS-kdtree-July-1-16-arxiv}

\begin{thebibliography}{10}

\bibitem{af}
P.~Afshani, L.~Arge, and K.G. Larsen.
\newblock Higher-dimensional orthogonal range reporting and rectangle stabbing
  in the pointer machine model.
\newblock In {\em Proceedings of the twenty-eighth annual symposium on
  Computational geometry}, pages 323--332. ACM, 2012.

\bibitem{pk}
P.~K. Agarwal.
\newblock {\em Range searching. In J. E. Goodman and J. O’Rourke, editors,
  CRC Handbook of Discrete and Computational Geometry}.
\newblock CRC Press, Inc, 2004.

\bibitem{als}
S.~Alstrup, G.~S. Brodal, and T.~Rauhe.
\newblock New data structures for orthogonal range searching.
\newblock In {\em Foundations of Computer Science, 2000. Proceedings. 41st
  Annual Symposium on}, pages 198--207. IEEE, 2000.

\bibitem{b5}
J.L. Bentley.
\newblock Multidimensional binary search tress used for associative searching.
\newblock {\em Communications of ACM}, 18(9):509--516, 1975.

\bibitem{b3}
J.L. Bentley.
\newblock Decomposable search problems.
\newblock {\em Information Processing Letters}, 8(5):5--9, June 1979.

\bibitem{b8}
J.L. Bentley.
\newblock Multidimensional binary search trees in database applications.
\newblock {\em IEEE Transactions on Software Engineering}, SE-5(4):333--340,
  1979.

\bibitem{b7}
J.L. Bentley.
\newblock Multidimensional divide and conquer.
\newblock {\em Communications of the ACM}, 23(4):214--229, April 1980.

\bibitem{bcko}
M.D. Berg, O.~Cheong, M.V. Kreveld, and M.~Overmars.
\newblock {\em Computational Geometry: algorithms and applications}.
\newblock Springer-Verlag, New York,USA, third edition, 2008.

\bibitem{ct}
T.M. Chan.
\newblock Persistent predecessor search and orthogonal point location on the
  word ram.
\newblock In {\em Proceedings of the Twenty-second Annual ACM-SIAM Symposium on
  Discrete Algorithms}, SODA '11, pages 1131--1145. SIAM, 2011.

\bibitem{clkm}
T.M. Chan, K.G. Larsen, and M.~P\u{a}tra\c{s}cu.
\newblock Orthogonal range searching on the ram, revisited.
\newblock In {\em Proceedings of the Twenty-seventh Annual Symposium on
  Computational Geometry}, SoCG '11, pages 1--10, New York, NY, USA, 2011. ACM.

\bibitem{mpc}
M.M.P. Crespo.
\newblock {\em Design, Analysis and Implementation of New Variants of
  Kd-trees}.
\newblock Master Thesis, Universitat Politecnica de Catalunya,Departament de
  Llenguatges i Sistemes Informatics, 2010.

\bibitem{ldjc}
L.~Devroye, J.~Jabbour, and C.~Zamora-Cura.
\newblock Squarish kd-trees.
\newblock {\em SIAM Journal of Computing}, 30:1678--1700, 2000.

\bibitem{w}
D.Willard.
\newblock New data structures for orthogonal queries.
\newblock {\em Harvard University}, TR:22--78, 1978.

\bibitem{ekh}
K.S. Easwarakumar and T.~Hema.
\newblock {BITS}-{T}ree-{A}n efficient data structure for segment storage and
  query processing.
\newblock {\em International Journal of Computers and Technology},
  11(10):3108--3116, December 2013.

\bibitem{mb}
M.G. Lamoureux and B.G. Nicolson.
\newblock Determinisitic skip lists for k-dimensional range search.
\newblock {\em Technical Report(TR95-098)}, pages 1--95, Novemeber 1995.

\bibitem{lw}
D.T. Lee and C.~K. Wong.
\newblock Worst-case analysis for region and partial region searches in
  multidimensional binary search trees and balanced quad trees.
\newblock {\em Acta Informatica}, pages 23--29, 1977.

\bibitem{nek}
Y.~Nekrich.
\newblock Orthogonal range searching in linear and almost-linear space.
\newblock {\em Computational Geometry}, 42(4):342--351, 2009.

\bibitem{nt}
S.~Nilsson and M.Tikkanen.
\newblock An experimental study of compression methods for dynamic tries.
\newblock {\em Algorithmica}, 33(1):19--33, 2002.

\bibitem{jo}
J.A. Orienstein.
\newblock Multidimensional tries used for associative searching.
\newblock {\em Information Processing Letters}, 14(4):150--157, June 1982.

\bibitem{ps}
F.P. Preparata and M.L. Shamos.
\newblock {\em Computational Geometry: An Introduction}.
\newblock Springer-Verlag, New York, 1985.

\bibitem{ns}
R.C.Nelson and H.Samet.
\newblock A consistent hierarchical representation for vector data.
\newblock In {\em Proceedings of the SIGGRAPH'86 Conference, Dallas},
  volume~20, pages 197--206, August 1986.

\bibitem{h2}
H.~Samet.
\newblock {\em Fundamentals of Multi-dimensional and Metric Data Structures}.
\newblock Academic Press, New York,USA, 1974.

\bibitem{h1}
H.~Samet.
\newblock {\em The Design and Analysis of Spatial Data Structures}.
\newblock Addison Wesley, 1990.

\bibitem{hth}
H.~Tropf and H.Herzog.
\newblock Multidimensional range search in dynamically balanced trees.
\newblock {\em Applied Informatics, Vieweg Verlag,Germany}, 2:71--77, 1981.

\bibitem{km}
M.J. van Kreveld and M.H. Overmars.
\newblock Divided k-d trees.
\newblock {\em Algorithmica}, 6:840--858, 1991.

\end{thebibliography}
\end{document}